\documentclass[11pt]{myclass}

\usepackage{mystyle}
\usepackage[outercaption]{sidecap}
\usepackage{amsmath}
\usepackage{graphics}
\usepackage{graphicx}

\newcommand{\Paragraph}[1]{\paragraph{\sffamily \textbf{#1}}}

\title{Subsampling in Smoothed Range Spaces}
\author{
Jeff M. Phillips \\ {\small\textsl{jeffp@cs.utah.edu}} \\  {\small  University of Utah} \and
Yan Zheng \\ {\small\textsl{yanzheng@cs.utah.edu}} \\  {\small University of Utah}
}
\date{}

\begin{document}
\maketitle

\begin{abstract}
We consider smoothed versions of geometric range spaces, so an element of the ground set (e.g. a point) can be contained in a range with a non-binary value in $[0,1]$.  Similar notions have been considered for kernels; we extend them to more general types of ranges.  We then consider approximations of these range spaces through $\eps$-nets and $\eps$-samples (aka $\eps$-approximations).  We characterize when size bounds for $\eps$-samples on kernels can be extended to these more general smoothed range spaces.  We also describe new generalizations for $\eps$-nets to these range spaces and show when results from binary range spaces can carry over to these smoothed ones.  
\end{abstract}

\section{Introduction}
\label{sec:intro}

This paper considers traditional sample complexity problems but adapted to when the range space (or function space) smoothes out its boundary.  This is important in various scenarios where either the data points or the measuring function is noisy.  Similar problems have been considered in specific contexts of functions classes with a $[0,1]$ range or kernel density estimates.  We extend and generalize various of these results, motivated by scenarios like the following.  

\begin{itemize}
\item[(S1)]
Consider maintaining a random sample of noisy spatial data points (say twitter users with geo-coordinates), and we want this sample to include a witness to every large enough event.  However, because the data coordinates are noisy we use a kernel density estimate to represent the density.  And moreover, we do not want to consider regions with a single or constant number of data points which only occurred due to random variations.  In this scenario, how many samples do we need to maintain?  

\item[(S2)]
Next consider a large approximate  (say high-dimensional image feature~\cite{PiCoDes11}) dataset, where we want to build a linear classifier.  Because the features are approximate (say due to feature hashing techniques), we model the classifier boundary to be randomly shifted using Gaussian noise.  How many samples from this dataset do we need to obtain a desired generalization bound?  

\item[(S3)]
Finally, consider one of these scenarios in which we are trying to create an informative subset of the enormous full dataset, but have the opportunity to do so in ways more intelligent than randomly sampling.  On such a reduced dataset one may want to train several types of classifiers, or to estimate the density of various subsets.  Can we generate a smaller dataset compared to what would be required by random sampling?  
\end{itemize}

The traditional way to study related sample complexity problems is through 
range spaces (a ground set $X$, and family of subsets $\c{A}$) and their associated dimension (e.g., VC-dimension~\cite{VC71}).  
We focus on a smooth extension of range spaces defined on a geometric ground set.  Specifically, consider the ground set $P$ to be a subset of points in $\b{R}^d$, and let $\c{A}$ describe subsets defined by some geometric objects, for instance a halfspace or a ball.  Points $p \in \b{R}^d$ that are inside the object (e.g., halfspace or ball) are typically assigned a value $1$, and those outside a value $0$. In our smoothed setting points near the boundary are given a value between $0$ and $1$, instead of discretely switching from $0$ to $1$.  

In learning theory these smooth range spaces can be characterized by more general notions called $P$-dimension~\cite{Pol90} (or Pseudo dimension) or $V$-dimension~\cite{Vap89} (or ``fat'' versions of these~\cite{ABCH97}) and can be used to learn real-valued functions for regression or density estimation, respectively.

In geometry and data structures, these smoothed range spaces are of interest in studying noisy data.   
Our work extends some recent work~\cite{JoshiKommarajuPhillips2011,Phillips2013} which examines a special case of our setting that maps to kernel density estimates, and matches or improves on related bounds for non-smoothed versions.  

\Paragraph{Main contributions.}
We next summarize the main contributions in this paper. 
\begin{itemize}
\item[$\bullet$]
We define a general class of \emph{smoothed range spaces} (Sec \ref{sec:smooth-RS}), with application to density estimation and noisy agnostic learning, and we show that these can inherit sample complexity results from \emph{linked} non-smooth range spaces (Corollary \ref{cor:inher-sample}).  
\item[$\bullet$] 
We define an $(\eps,\tau)$-net for a smoothed range space (Sec \ref{sec:et-net}).  We show how this can inherit sampling complexity bounds from \emph{linked} non-smooth range spaces (Theorem \ref{thm:inher-net}), and we relate this to non-agnostic density estimation and hitting set problems.  
\item[$\bullet$]
We provide discrepancy-based bounds and constructions for $\eps$-samples on smooth range spaces requiring significantly fewer points than uniform sampling approaches (Theorems \ref{thm:epssample} and \ref{thm:epssample-d}), and also smaller than discrepancy-based bounds on the linked binary range spaces.  
These are useful for batched active learning, where a prespecified batch of (not uniform at random) samples can be then asked for labels to be used for learning.  
\end{itemize}

\section{Definitions and Background}
\label{sec:definitions}

Recall that we will focus on geometric range spaces $(P,\c{A})$ where the ground set $P \subset \b{R}^d$ and the family of ranges $\c{A}$ are defined by geometric objects.
It is common to approximate a range space in one of two ways, as an $\eps$-sample (aka $\eps$-approximation) or an $\eps$-net.  
An \emph{$\eps$-sample} for a range space $(P,\c{A})$ is a subset $Q \subset P$ such that 
\[
\max_{A \in \c{A}} \left| \frac{|A \cap P|}{|P|} - \frac{|Q \cap A|}{|Q|} \right| \leq \eps.
\]
An \emph{$\eps$-net} of a range space $(P,\c{A})$ is a subset $Q \subset P$ such that 
\[
\textrm{for all } A \in \c{A} \textrm{ such that } \frac{|P \cap A|}{|P|} \geq \eps \textrm{ then } A \cap Q \neq \emptyset.
\]
Given a range space $(P,\c{A})$ where $|P| =m$, then $\pi_\c{A}(m)$ describes the maximum number of possible distinct subsets of $P$ defined by some $A \in \c{A}$.  If we can bound, $\pi_{\c{A}}(m) \leq C m^\nu$ for absolute constant $C$, then $(P,\c{A})$ is said to have \emph{shatter dimension} $\nu$.  
For instance the shatter dimension of $\c{H}$ halfspaces in $\b{R}^d$ is $d$, and for $\c{B}$ balls in $\b{R}^d$ is $d+1$.  
For a range space with shatter dimension $\nu$, a random sample of size $O((1/\eps^2)(\nu+\log(1/\delta)))$ is an $\eps$-sample with probability at least $1-\delta$~\cite{VC71,LLS01}, and a random sample of size $O((\nu/\eps) \log(1/\eps \delta))$ is an $\eps$-net with probability at least $1-\delta$~\cite{HW87,Path95}.  

An $\eps$-sample $Q$ is sufficient for agnostic learning with generalization error $\eps$, where the best classifier might misclassify some points.  An $\eps$-net $Q$ is sufficient for non-agnostic learning with generalization error $\eps$, where the best classifier is assumed to have no error on $P$.  

The size bounds can be made deterministic and slightly improved for certain cases.  
An $\eps$-sample $Q$ can be made of size $O(1/\eps^{2\nu/(\nu+1)})$~\cite{Mat95} and this bound can be no smaller~\cite{Mat99} in the general case.  For balls $\c{B}$ in $\b{R}^d$ which have shatter-dimension $\nu=d+1$, this can be improved to $O(1/\eps^{2d/(d+1)} \log^{d/(d+1)}(1/\eps))$~\cite{Bec87,Mat99}, and the best known lower bound is $\Omega(1/\eps^{2d/(d+1)})$.  For axis-aligned rectangles $\c{R}$ in $\b{R}^d$ which have shatter-dimension $\nu=2d$, this can be improved to $O((1/\eps) \log^{d+1/2} (1/\eps))$~\cite{Lar11}.  

For $\eps$-nets, the general bound of $O((\nu/\eps) \log(1/\eps))$ can also be made deterministic~\cite{Mat95}, and for halfspaces in $\b{R}^4$ the size must be at least $\Omega((1/\eps) \log (1/\eps))$~\cite{PT13}.  But for halfspaces in $\b{R}^3$ the size can be $O(1/\eps)$~\cite{MSW90,HKSS14}, which is tight.  By a simple lifting, this also applies for balls in $\b{R}^2$.  For other range spaces, such as axis-aligned rectangles in $\b{R}^2$, the size bound is $\Theta((1/\eps) \log\log(1/\eps))$~\cite{AES10,PT13}.

\subsection{Kernels}
\label{sec:kernels}
A \emph{kernel} is a bivariate similarity function $K : \b{R}^d \times \b{R}^d \to \b{R}^+$, which can be normalized so $K(x,x) = 1$ (which we assume through this paper).  Examples include 
ball kernels ($K(x,p) = \{1$ if  $\|x-p\| \leq 1$ and $0$ otherwise\}), 
triangle kernels ($K(x,p) = \max \{0, 1- \|x-p\|\}$), 
Epanechnikov kernels ($K(x,p) = \max \{0, 1- \|x-p\|^2\}$), and 
Gaussian kernels ($K(x,p) = \exp(-\|x-p\|^2)$, which is reproducing).  
In this paper we focus on symmetric, shift invariant kernels which depend only on $z = \|x-p\|$, and can be written as a single parameter function $K(x,p) = k(z)$; these can be parameterized by a single bandwidth (or just width) parameter $w$ so $K(x,p) = k_w(\|x-p\|/w)$.  

Given a point set $P \subset \b{R}^d$ and a kernel, a \emph{kernel density estimate} $\kde_P$ is the convolution of that point set with $K$.  For any $x \in \b{R}^d$ we define $\kde_P(x) = \frac{1}{|P|} \sum_{p \in P} K(x,p)$.  

A \emph{kernel range space}~\cite{JoshiKommarajuPhillips2011,Phillips2013} $(P,\c{K})$ is an extension of the combinatorial concept of a range space $(P,\c{A})$ (or to distinguish it we refer to the classic notion as a \emph{binary range space}).  
It is defined by a point set $P \subset \b{R}^d$ and a kernel $K$.  An element $K_x$ of $\c{K}$ is a kernel $K(x,\cdot)$ applied at point $x \in \b{R}^d$; it assigns a value in $[0,1]$ to each point $p \in P$ as $K(x,p)$.  
If we use a ball kernel, then each value is exactly $\{0,1\}$ and we recover exactly the notion of a binary range space for geometric ranges defined by balls.  

The notion of an $\eps$-kernel sample~\cite{JoshiKommarajuPhillips2011} extends the definition of $\eps$-sample.  It is a subset $Q \subset P$ such that 
\[
\max_{x \in \b{R}^d} \left| \kde_P(x) - \kde_Q(x) \right| \leq \eps.
\]

A binary range space $(P,\c{A})$ is \emph{linked} to a kernel range space $(P,\c{K})$ if the set $\{p \in P \mid K(x,p) \geq \tau\}$ is equal to $P \cap A$ for some $A \in \c{A}$, for any threshold value $\tau$.  
\cite{JoshiKommarajuPhillips2011} showed that an $\eps$-sample of a linked range space $(P,\c{A})$ is also an $\eps$-kernel sample of a corresponding kernel range space $(P,\c{K})$.  Since all range spaces defined by symmetric, shift-invariant kernels are linked to range spaces defined by balls, they inherit all $\eps$-sample bounds, including that random samples of size $O((1/\eps^2) (d + \log(1/\delta))$ provide an $\eps$-kernel sample with probability at least $1-\delta$.  
Then \cite{Phillips2013} showed that these bounds can be improved through discrepancy-based methods to $O(((1/\eps)\sqrt{\log(1/\eps\delta)})^{2d/(d+2)})$, which is $O((1/\eps) \sqrt{\log(1/\eps\delta)})$ in $\b{R}^2$.  

A more general concept has been studied in learning theory on real-valued functions, where a function $f$ as a member of a function class $\c{F}$ describes a mapping from $\b{R}^d$ to $[0,1]$ (or more generally $\b{R}$).  
A kernel range space where the linked binary range space has bounded shatter-dimension $\nu$ is said to have bounded V-dimension \cite{Vap89} (see \cite{ABCH97}) of $\nu$.  Given a ground set $X$, then for $(X,\c{F})$ this describes the largest subset $Y$ of $X$ which can be shattered in the following sense.  Choose any value $s \in [0,1]$ for all points $y \in Y$, and then for each subset of $Z \subset Y$ there exists a function $f \in \c{F}$ so $f(y) > s$ if $y \in Z$ and $f(y) < s$ if $y \notin Z$.  
The best sample complexity bounds for ensuring $Q$ is an $\eps$-sample of $P$ based on V-dimension are derived from a more general sort of dimension (called a P-dimension \cite{Pol90} where in the shattering definition, each $y$ may have a distinct $s(y)$ value) requires $|Q| = O((1/\eps^2) (\nu + \log(1/\delta)))$ \cite{LLS01}.  
As we will see, these V-dimension based results are also general enough to apply to the to-be-defined smooth range spaces. 

\section{New Definitions}
\label{sec:newdef}
In this paper we extend the notion of a kernel range spaces to other \emph{smoothed range spaces} that are ``linked'' with common range spaces, e.g., halfspaces.  These inherent the construction bounds through the linking result of \cite{JoshiKommarajuPhillips2011}, and we show cases where these bounds can also be improved.  We also extend the notion of $\eps$-nets to kernels and smoothed range spaces, and showing linking results for these as well.

\subsection{Smoothed Range Spaces}
\label{sec:smooth-RS}
Here we will define the primary smoothed combinatorial object we will examine, starting with halfspaces, and then generalizing.  
Let $\c{H}_w$ denote the family of \emph{smoothed halfspaces with width parameter $w$}, and let $(P,\c{H}_w)$ be the associated smoothed range space where $P \subset \b{R}^d$.
Given a point $p \in P$, then smoothed halfspace $h \in \c{H}_w$ maps $p$ to a value $v_h(p) \in [0,1]$ (rather than the traditional $\{0,1\}$ in a binary range space).  

We first describe a specific mapping to the function value $v_h(p)$ that will be sufficient for the development of most of our techniques.  Let $F$ be the $(d-1)$-flat defining the boundary of halfspace $h$.  Given a point $p \in \b{R}^d$, let $p_F = \arg \min_{q \in F} \|p-q\|$ describe the point on $F$ closest to $p$.  Now we define
\[
v_{h,w}(p) = 
\begin{cases}
1                                                                 & p \in h \textrm{ and } \|p-p_F\| \geq w
\\
\frac{1}{2} + \frac{1}{2} \frac{\|p-p_F\|}{w} & p \in h \textrm{ and } \|p-p_F\| < w
\\
\frac{1}{2} - \frac{1}{2} \frac{\|p-p_F\|}{w} & p \notin h \textrm{ and }  \|p - p_F\| < w
\\
0                                                                & p \notin h \textrm{ and } \|p - p_F\| \geq w.
\end{cases}
\]
These points within a slab of width $2w$ surrounding $F$ can take on a value between $0$ and $1$, where points outside of this slab revert back to the binary values of either $0$ or $1$.  

We can make this more general using a shift-invariant kernel $k(\|p-x\|) = K(p,x)$, where $k_w(\|p-x\|) = k(\|p-x\|/w)$ allows us to parameterize by $w$.   Define $v_{h,w}(p)$ as follows.  
\[
v_{h,w}(p) = 
\begin{cases}
1 - \frac{1}{2} k_w(\|p-p_F\|)                      & p \in h 
\\
\frac{1}{2} k_w(\|p-p_F\|)                          & p \notin h. 
\end{cases}
\]
For brevity, we will omit the $w$ and just use $v_h(p)$ when clear.  
These definitions are equivalent when using the triangle kernel.  But for instance we could also use a Epanechnikov kernel or Gaussian kernel.  Although the Gaussian kernel does not satisfy the restriction that only points in the width $2w$ slab take non $\{0,1\}$ values, we can use techniques from \cite{Phillips2013} to extend to this case as well.  This is illustrated in Figure \ref{fig:smhalf}.  
Another property held by this definition which we will exploit is that the slope $\varsigma$ of these kernels is bounded by $\varsigma = O(1/w) = c/w$, for some constant $c$; the constant $c = 1/2$ for triangle and Gaussian, and $c = 1$ for Epanechnikov.   

\begin{figure}[t!]
    \begin{minipage}[t]{0.3\linewidth}
   \includegraphics[width=3in]{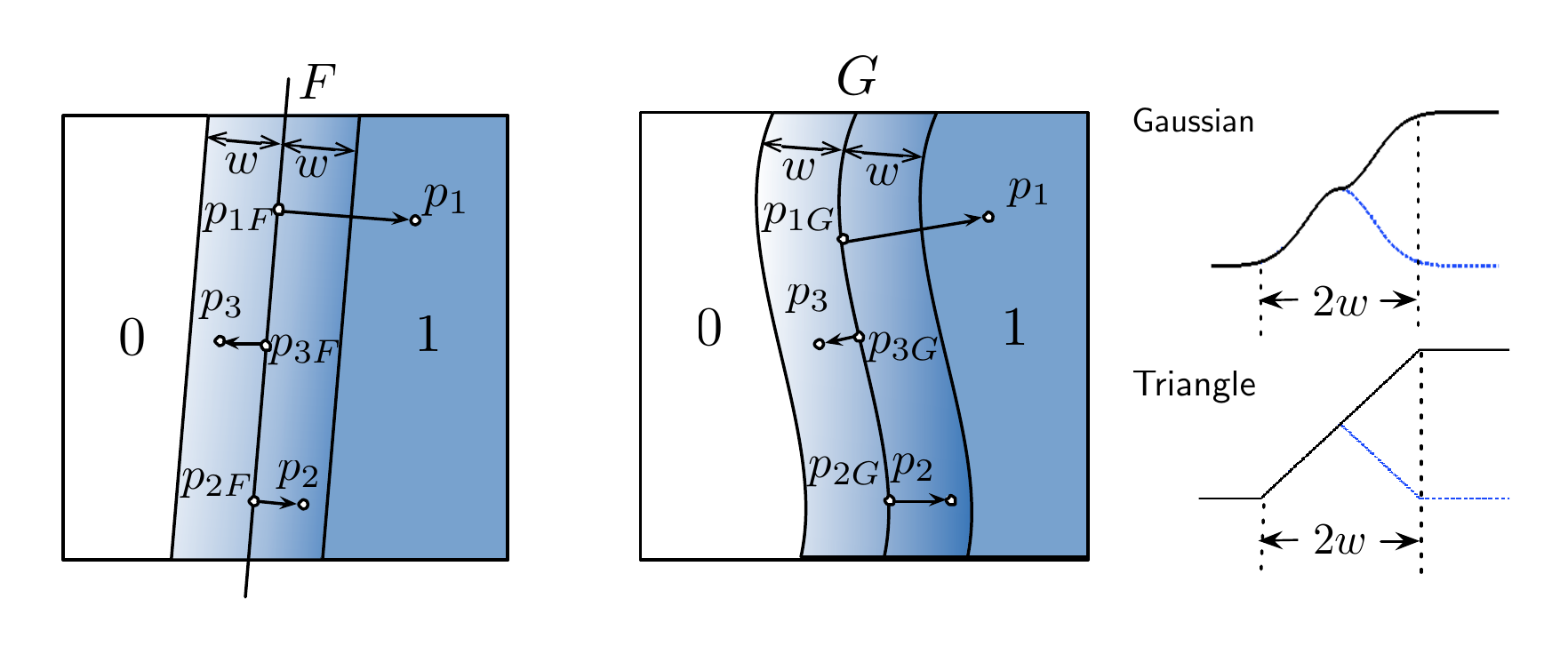}  
    \end{minipage}

    \hfill
    \begin{minipage}[t]{0.35\linewidth}  
    \vspace{-30mm}
	\begin{tabular}{l*{6}{c}r}
	                   & $\|p-p_F\|$ & $?(p \in h)$    & $v_h(p)$ \\\hline
	$p_1$ 	    & $3w/2$        & \textsc{true} & 1    \\
	$p_2$         & $3w/4$        & \textsc{true}  & 7/8  \\
	$p_3$         & $w/2$          & \textsc{false} & 1/4 \\
	\end{tabular}
	\end{minipage}

\vspace{-3.5mm}
\caption{Illustration of the smoothed halfspace, and smoothed polynomial surface, with function value of three points $\{p_1,p_2, p_3\}$ defined using a triangle kernel. }
  \label{fig:smhalf}
\end{figure}

Finally, we can further generalize this by replacing the flat $F$ at the boundary of $h$ with a polynomial surface $G$.  The point $p_G = \arg \min_{q \in G} \|p -q\|$ replaces $p_F$ in the above definitions.  Then the slab of width $2w$ is replaced with a curved volume in $\b{R}^d$; see Figure \ref{fig:smhalf}.  
For instance, if $G$ defines a circle in $\b{R}^d$, then $v_h$ defines a disc of value $1$, then an annulus of width $2w$ where the function value decreases to $0$.  
Alternatively, if $G$ is a single point, then we essentially recover the kernel range space, except that the maximum height is $1/2$ instead of $1$.  
We will prove the key structural results for polynomial curves in Section \ref{sec:match}, but otherwise focus on halfspaces to keep the discussion cleaner.  The most challenging elements of our results are all contained in the case with $F$ as a $(d-1)$-flat.

\subsection{$\eps$-Sample in a Smoothed Range Space}
\label{sec:e-sample-SRS}
It will be convenient to extend the notion of a kernel density estimate to these smoothed range space.  
A \emph{smoothed density estimate} $\sde_P$ is defined for any $h \in \c{H}_w$ as 
\[
\sde_P(h) = \frac{1}{|P|} \sum_{p \in P} v_h(p).
\]
An \emph{$\eps$-sample $Q$ of a smoothed range space} $(P, \c{H}_w)$ is a subset $Q \subset P$ such that
\[
\max_{h \in \c{H}_w} \left| \sde_P(h) - \sde_Q(h) \right| \leq \eps.
\]
Given such an $\eps$-sample $Q$, we can then consider a subset $\bar{\c{H}}_w$ of $\c{H}_w$ with bounded integral (perhaps restricted to some domain like a unit cube that contains all of the data $P$).  If we can learn the smooth range $\hat h = \arg \max_{h \in \bar{\c{H}}_w} \sde_Q(h)$, then we know
$\sde_P(h^*) - \sde_Q(\hat h) \leq \eps$, where $h^* = \arg \max_{h \in \bar{\c{H}}_w} \sde_P(h)$, since $\sde_Q(\hat h) \geq \sde_Q(h^*) \geq \sde_P(h^*) - \eps$.  Thus, such a set $Q$ allows us to learn these more general density estimates with generalization error $\eps$.  

We can also learn smoothed classifiers, like scenario (S2) in the introduction, with generalization error $\eps$, by giving points in the negative class a weight of $-1$; this requires separate $(\eps/2)$-samples for the negative and positive classes.  

\subsection{$(\eps, \tau)$-Net in a Smoothed Range Space}
\label{sec:et-net}
We now generalize the definition of an $\eps$-net.  Recall that it is a subset $Q \subset P$ such that $Q$ ``hits'' all large enough ranges ($|P \cap A|/|P| \geq \eps$).  However, the notion of ``hitting'' is now less well-defined since a point $q \in Q$ may be in a range but with value very close to $0$; if a smoothed range space is defined with a Gaussian or other kernel with infinite support, any point $q$ will have a non-zero value for \emph{all} ranges!  
Hence, we need to introduce another parameter $\tau \in (0, \eps)$, to make the notion of hitting more interesting  in this case.  

A subset $Q \subset P$ is an \emph{$(\eps,\tau)$-net of smoothed range space} $(P,\c{H}_w)$ if for any smoothed range $h \in \c{H}_w$ such that $\sde_P(h) \geq \eps$, then there exists a point $q \in Q$ such that $v_h(q) \geq \tau$.

The notion of $\eps$-net is closely related to that of hitting sets.  
A \emph{hitting set} of a binary range space $(P,\c{A})$ is a subset $Q \subset P$ so every $A \in \c{A}$ (not just the large enough ones) contains some $q \in Q$.  
To extend these notions to the smoothed setting, we again need an extra parameter $\tau \in (0,\eps)$, and also need to only consider large enough smoothed ranges, since there are now an infinite number even if $P$ is finite.  
A subset $Q \subset P$ is an \emph{$(\eps,\tau)$-hitting set of smoothed range space} $(P,\c{H}_w)$ if for any $h \in \c{H}_w$ such that $\sde_P(h) \geq \eps$, then $\sde_Q(h) \geq \tau$.  

In the binary range space setting, an $\eps$-net $Q$ of a range space $(P,\c{A})$ is sufficient to learn the best classifier on $P$ with generalization error $\eps$ in the non-agnostic learning setting, that is assuming a perfect classifier exists on $P$ from $\c{A}$.  In the density estimation setting, there is not a notion of a perfect classifier, but if we assume some other properties of the data, the $(\eps,\tau)$-net will be sufficient to recover them.  
For instance, consider (like scenario (S1) in the introduction) that $P$ is a discrete distribution so for some ``event'' points $p \in P$, there is at least an $\eps$-fraction of the probability distribution describing $P$ at $p$ (e.g., there are more than $\eps |P|$ points very close to $p$).  In this setting, we can recover the location of these points since they will have probability at least $\tau$ in the $(\eps,\tau)$-net $Q$.

\section{Linking and Properties of $(\eps, \tau)$-Nets}
\label{sec:net}

First we establish some basic connections between $\eps$-sample, $(\eps,\tau)$-net, and $(\eps,\tau)$-hitting set in smoothed range spaces.  In binary range spaces an $\eps$-sample $Q$ is also an $\eps$-net, and a hitting set is also an $\eps$-net; we show a similar result here up to the covering constant $\tau$.  

\begin{lemma}
\label{lem:hit-net}
For a smoothed range space $(P, \c{H}_w)$ and $0 < \tau < \eps < 1$, an $(\eps,\tau)$-hitting set $Q$ is also an $(\eps,\tau)$-net of $(P,\c{H}_w)$.  
\end{lemma}
\begin{proof}
The $(\eps,\tau)$-hitting set property establishes for all $h \in \c{H}_w$ with $\sde_P(h) \geq \eps$, then also $\sde_Q(h) \geq \tau$.  Since $\sde_Q(h) = \frac{1}{|Q|}\sum_{q \in Q} v_h(q)$ is the average over all points $q \in Q$, then it implies that at least one point also satisfies $v_h(q) \geq \tau$.  Thus $Q$ is also an $(\eps,\tau)$-net.  
 \end{proof}

In the other direction an $(\eps,\tau)$-net is not necessarily an $(\eps,\tau)$-hitting set since the $(\eps,\tau)$-net $Q$ may satisfy a smoothed range $h \in \c{H}_w$ with a single point $q \in Q$ such that $v_h(q) \geq \tau$, but all others $q' \in Q\setminus \{q\}$ having $v_h(q') \ll \tau$, and thus $\sde_Q(h) < \tau$.

\begin{theorem}
\label{thm:net} 
For $0 < \tau < \eps < 1$, an $(\eps - \tau)$-sample $Q$ in smoothed range space $(P,\c{H}_w)$ is an $(\eps, \tau)$-hitting set in $(P,\c{H}_w)$, and thus also an $(\eps,\tau)$-net of $(P,\c{H}_w)$.  
\end{theorem}
\begin{proof}
Since $Q$ is the $(\eps - \tau)$-sample in the smoothed range space, for any smoothed range $h \in \c{H}_w$ we have $| \sde_P(h) - \sde_Q(h) | \leq \eps - \tau$.  
We consider the upper and lower bound separately.  

If $\sde_P(h) \ge \eps$, when $\sde_P(h) \ge \sde_Q(h)$, we have
\[
\sde_Q(h)  \geq \sde_P(h) - (\eps -\tau) \geq \eps - (\eps - \tau) = \tau.  
\]
And more simply when $\sde_Q(h) \geq \sde_P(h)$ and $\sde_P(h) \geq \eps \geq \tau$, then $\sde_Q(h) \geq \tau$.  
Thus in both situations, $Q$ is an $(\eps, \tau)$-hitting set of $(P,\c{H}_w)$.
And then by Lemma \ref{lem:hit-net} $Q$ is also an $(\eps,\tau)$-net of $(P,\c{H}_w)$.  
 \end{proof}

\subsection{Relations between Smoothed Range Spaces and Linked Binary Range Spaces}
Consider a smoothed range space $(P,\c{H}_w)$, and for one smoothed range $h \in \c{H}_w$, examine the range boundary $F$ (e.g. a $(d-1)$-flat, or polynomial surface) along with a symmetric, shift invariant kernel $K$ that describes $v_h$.  The \emph{superlevel} set $(v_h)^\tau$ is all points $x \in \b{R}^d$ such that $v_h(x) \geq \tau$.  
Then recall a smoothed range space $(P,\c{H}_w)$ is \emph{linked} to a binary range space $(P,\c{A})$ if every set $\{p \in P \mid v_h(p) \geq \tau\}$ for any $h \in \c{H}_w$ and any $\tau > 0$, is exactly the same as some range $A \cap P$ for $A \in \c{A}$.  
For smoothed range spaces defined by halfspaces, then the linked binary range space is also defined by halfspaces.  For smoothed range spaces defined by points, mapping to kernel range spaces, then the linked binary range spaces are defined by balls.  

Joshi \etal~\cite{JoshiKommarajuPhillips2011} established that given a kernel range space $(P,\c{K})$, a linked binary range space $(P,\c{A})$, and an $\eps$-sample $Q$ of $(P,\c{A})$, then $Q$ is also an $\eps$-kernel sample of $(P,\c{K})$.  
An inspection of the proof reveals the same property holds directly for smoothed range spaces, as the only structural property needed is that all points $p \in P$, as well as all points $q \in Q$, can be sorted in decreasing function value $K(p,x)$, where $x$ is the center of the kernel.  For smoothed range space, this can be replaced with sorting by $v_h(p)$.  

\begin{corollary}[\cite{JoshiKommarajuPhillips2011}] \label{cor:inher-sample}
Consider a smoothed range space $(P,\c{H}_w)$, a linked binary range space $(P,\c{A})$, and an $\eps$-sample $Q$ of $(P,\c{A})$ with $\eps \in (0,1)$.  Then $Q$ is an $\eps$-sample of $(P,\c{H}_w)$.  
\end{corollary}

We now establish a similar relationship to $(\eps,\tau)$-nets of smoothed range spaces from $(\eps-\tau)$-nets of linked binary range spaces.  

\begin{theorem} \label{thm:inher-net}
Consider a smoothed range space $(P,\c{H}_w)$, a linked binary range space $(P,\c{A})$, and an $(\eps-\tau)$-net $Q$ of $(P,\c{A})$ for $0 < \tau < \eps < 1$.  Then $Q$ is an $(\eps,\tau)$-net of $(P,\c{H}_w)$.  
\end{theorem}
\begin{proof}
Let $|P| = n$.  
Then since $Q$ is an $(\eps-\tau)$-net of $(P, \c{A})$, for any range $A \in \c{A}$, if $|P \cap A| \geq (\eps - \tau)n$, then $Q \cap A \neq \emptyset$.

Suppose $h \in \c{H}_w$ has $\sde_P(h) \geq \eps$ and we want to establish that $\sde_Q(h) \geq \tau$.  
Let $A \in \c{A}$ be the range such that $(\eps-\tau)n$ points with largest $v_h(p_i)$ values are exactly the points in $A$.  We now partition $P$ into three parts
(1) let $P_1$  be the $(\eps - \tau)n - 1$ points with largest $v_h$ values, 
(2) let $y$ be the point in $P$ with $(\eps-\tau)n$th largest $v_h$ value, and
(3) let $P_2$ be the remaining $n - n(\eps-\tau)$ points.  
Thus for every $p_1 \in P_1$ and every $p_2 \in P_2$ we have $v_h(p_2) \leq v_h(y) \leq v_h(p_1) \leq 1$.  

Now using our assumption $n \cdot \sde_P(h) \geq n \eps$ we can decompose the sum
\[
n \cdot \sde_P(h) = \sum_{p_1 \in P_1} v_h(p_1) + v_h(y) + \sum_{p_2 \in P_2} v_h(p_2) \geq n\eps, 
\]
and hence using upper bounds $v_h(p_1) \leq 1$ and $v_h(p_2) \leq v_h(y)$, 
\begin{align*}
v_h(y) 
&\geq 
n\eps - \sum_{p_1 \in P_1} v_h(p_1) - \sum_{p_2 \in P_2} v_h(p_2)
\\ &\geq
n \eps - (n (\eps -\tau) -1) \cdot 1 - (n-n(\eps-\tau)) v_h(y).  
\end{align*}
Solving for $v_h(y)$ we obtain
\[
v_h(y) 
\geq 
\frac{n \tau +1}{n - n(\eps-\tau) + 1} 
\geq 
\frac{n \tau}{n - n(\eps - \tau)} 
\geq
\frac{n \tau}{n} 
= 
\tau.
\]

Since $(P,\c{A})$ is linked to $(P,\c{H}_w)$, there exists a range $A \in \c{A}$ that includes precisely $P_1 \cup y$ (or more points with the same $v_h(y)$ value as $y$).  Because $Q$ is an $(\eps-\tau)$-net of $(P, \c{A})$,  $Q$ contains at least one of these points, lets call it $q$.  Since all of these points have function value $v_h(p) \geq v_h(y) \geq \tau$, then $v_h(q) \geq \tau$.  Hence $Q$ is also an $(\eps,\tau)$-net of $(P,\c{H}_w)$, as desired.  
 \end{proof}

This implies that if $\tau \leq c \eps$ for any constant $c < 1$, then creating an $(\eps,\tau)$-net of a smoothed range space, with a known linked binary range space, reduces to computing an $\eps$-net for the linked binary range space.  For instance any linked binary range space with shatter-dimension $\nu$ has an $\eps$-net of size $O(\frac{\nu}{\eps} \log \frac{1}{\eps})$, including halfspaces in $\b{R}^d$ with $\nu = d$ and balls in $\b{R}^d$ with $\nu = d+1$; hence there exists $(\eps,\eps/2)$-nets of the same size.  For halfspaces in $\b{R}^2$ or $\b{R}^3$ (linked to smoothed halfspaces) and balls in $\b{R}^2$ (linked to kernels), the size can be reduced to $O(1/\eps)$~\cite{MSW90,HKSS14,pyrga2008new}.

\section{Min-Cost Matchings within Cubes}
\label{sec:match}
Before we proceed with our construction for smaller $\eps$-samples for smoothed range spaces, we need to prepare some structural results about min-cost matchings.  
Following some basic ideas from \cite{Phillips2013}, these matchings will be used for discrepancy bounds on smoothed range spaces in Section \ref{sec:smhalf}.

In particular, we analyze some properties of the interaction of a min-cost matching $M$ and some basic shapes (\cite{Phillips2013} considered only balls).  Let $P \subset \b{R}^d$ be a set of $2n$ points.  A \emph{matching} $M(P)$ is a decomposition of $P$ into $n$ pairs $\{p_i,q_i\}$ where $p_i,q_i \in P$ and each $p_i$ (and $q_i$) is in exactly one pair.  A \emph{min-cost matching} is the matching $M$ that minimizes $\s{cost}_1(M,P) = \sum_{i=1}^n \|p_i - q_i\|$.  
The min-cost matching can be computed in $O(n^3)$ time by \cite{Edm65} (using an extension of the Hungarian algorithm from the bipartite case).  In $\b{R}^2$ it can be calculated in $O(n^{3/2} \log^5 n)$ time~\cite{Var98}. 

Following \cite{Phillips2013}, again we will base our analysis on a result of \cite{BE93} which says that if $P \subset [0,1]^d$ (a unit cube) then for $d$ a constant,
$\s{cost}_d(M,P) = \sum_{i=1}^n \|p_i - q_i\|^d = O(1)$, 
where $M$ is the min-cost matching.  We make no attempt to optimize constants, and assume $d$ is constant.  

One simple consequence, is that if $P$ is contained in a $d$-dimensional cube of side length $\ell$, then 
$\s{cost}_d(M,P) = \sum_{i=1}^n \|p_i - q_i\|^d = O(\ell^d)$.

We are now interested in interactions with a matching $M$ for $P$ in a $d$-dimensional cube of side length $\ell$ $\s{C}_{\ell,d}$ (call this shape an \emph{$(\ell,d)$-cube}), and more general objects; in particular 
$C_w$ a $(w,d)$-cube and,
$S_w$ a slab of width $2w$, 
both restricted to be within $\s{C}_{\ell,d}$.  
Now for such an object $O_w$ (which will either be $C_w$ or $S_w$) and an edge $\{p,q\}$ where line segment $\overline{pq}$ intersects $O_w$ define point $p_B$ (resp. $q_B$) as the point on segment $\overline{pq}$ inside $O_w$ closest to $p$ (resp. $q$).  Note if $p$ (resp. $q$) is inside $O$ then $p_B = p$ (resp. $q_B= q$), otherwise it is on the boundary of $O_w$.  
For instance, see $C_{20w}$ in Figure \ref{fig:w2cost}.  

\begin{wrapfigure}{r}{1.3in}
\vspace{-0.2in}
\includegraphics[width=1.3in]{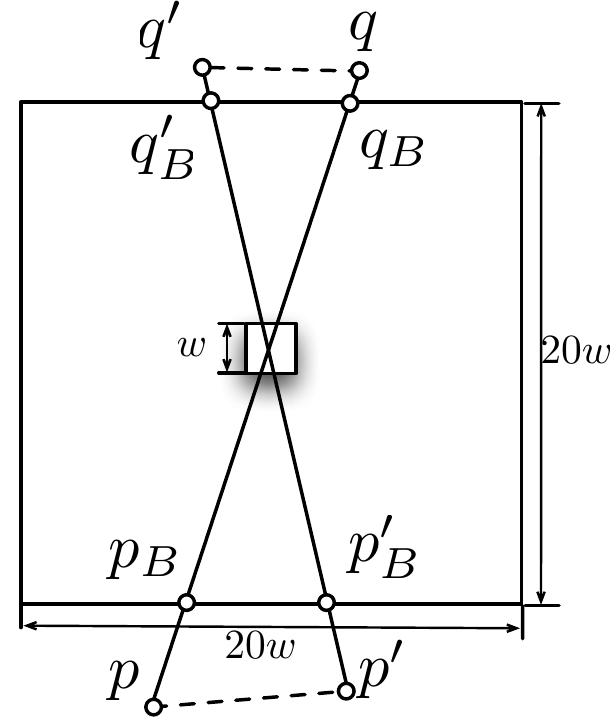}

\vspace{-.2in}

\caption{(T3) edges\label{fig:w2cost}}

\vspace{-.5in}
\end{wrapfigure}

Define the \emph{length} of a matching $M$ restricted to an object $O_w \subset \b{R}^d$ as 
\[
\rho(O_w,M) = \sum_{(q,p) \in M} \min \left\{ (2w)^d, \|p_B - q_B\|^d\right\}.
\]
Note this differs from a similar definition by \cite{Phillips2013} since that case did not need to consider when both $p$ and $q$ were both outside of $O_w$, and did not need the $\min \{ (2w)^d, \ldots \}$ term because all objects had diameter $2$.

\begin{lemma}
\label{lem:match-cube}
Let $P \subset \s{C}_{\ell,d}$, where $d$ is constant, and $M$ be its min-cost matching.
For any $(w,d)$-cube $C_w \subset \s{C}_{\ell,d}$ we have
$\rho(C_w, M) = O(w^d)$.
\end{lemma}

\begin{proof}
We cannot simply apply the result of \cite{BE93} since we do not restrict that $P \subset C_w$.  We need to consider cases where either $p$ or $q$ or both are outside of $C_w$.  As such, we have three types of edges we consider, based on a cube $C_{20w}$ of side length $20w$ and with center the same as $C_w$.  
\begin{enumerate}
\item[(T1)] Both endpoints are within $C_{20w}$ of edge length at most $\sqrt{d}20w$.
\item[(T2)] One endpoint is in $C_w$, the other is outside $C_{20w}$.
\item[(T3)] Both endpoints are outside $C_{20w}$.
\end{enumerate}

For all (T1) edges, the result of Bern and Eppstein can directly bound their contribution to $\rho(C_w,M)$ as $O(w^d)$ (scale to a unit cube, and rescale).  
For all (T2) edges, we can also bound their contribution to $\rho(C_w,M)$ as $O(w^d)$, by extending an analysis of \cite{Phillips2013} when both $C_w$ and $C_{20w}$ are similarly proportioned balls.  This analysis shows there are $O(1)$ such edges.  

We now consider the case of (T3) edges, restricting to those that also intersect $C_w$.  We argue there can be at most $O(1)$ of them.  In particular consider two such edges $\{p,q\}$ and $\{p',q'\}$, and their mappings to the boundary of $C_{20w}$ as $p_B, q_B, p'_B, q'_B$; see Figure \ref{fig:w2cost}.  
If $\|p_B - p'_B\| \leq 10 w$ and $\|q_B - q'_B\| \leq 10w$, then we argue next that this cannot be part of a min-cost matching since $\|p-p'\| + \|q-q'\| < \|p-q\| + \|p'-q'\|$, and it would be better to swap the pairing.  Then it follows from the straight-forward net argument below that there can be at most $O(1)$ such pairs.  

We first observe that 
$\|p_B-p'_B\| + \|q_B-q'_B\| \leq 10w + 10w < 20w + 20w \leq \|p_B-q_B\| + \|p'_B-q'_B\|$.  
Now we can obtain our desired inequality using that 
$\|p-q\| = \|p - p_B\| + \|p_B - q_B\| + \|q_B - q\|$ (and similar for $\|p'-q'\|$) 
and that 
$\|p-p'\| \leq \|p-p_B\| + \|p_B - p'_B\| + \|p'_B - p'\|$ by triangle inequality (and similar for $\|q-q'\|$).  

Next we describe the net argument that there can be at most $O(d^2\cdot 2^{2d}) = O(1)$ such pairs with $\|p_B - p'_B\| > 10w$ and $\|q_B - q'_B\| >10w$.  
First place a $5w$-net $\c{N}_f$ on each $(d-1)$-dimensional face $f$ of $C_{20w}$ so that any point $x \in f$ is within $5w$ of some point $\eta \in \c{N}_f$.  We can construct $\c{N}_f$ of size $O(2^d)$ with a simple grid.  Then let $\c{N} = \bigcup_f \c{N}_f$ as the union of the nets on each face; its size is $O(d \cdot 2^d)$.  
Now for any point $p \notin C_{20w}$ let $\eta(p) = \arg\min_{\eta \in \c{N}} \|p_B - \eta\|$ be the closest point in $\c{N}$ to $p_B$.  If two points $p$ and $p'$ have $\eta(p) = \eta(p')$ then $\|p-p'\| \leq 10w$.  
Hence there can be at most $O((d \cdot 2^d)^2)$ edges with $\{p,q\}$ mapping to unique $\eta(p)$ and $\eta(q)$ if no other edge $\{p',q'\}$ has $\|p_B-p'_B\| \leq 10w$ and $\|q_B - q_B'\| \leq 10w$.  

Concluding, there can be at most $O(d^2 \cdot 2^{2d}) = O(1)$ edges in $M$ of type (T3), and the sum of their contribution to $\rho(C_w,M)$ is at most $O(w^d)$, completing the proof.  
 \end{proof}

\begin{lemma}
\label{lem:match-slab}
Let $P \subset \s{C}_{\ell,d}$, where $d$ is constant, and let $M$ be its min-cost matching.  
For any width $2w$ slab $S_w$ restricted to $\s{C}_{\ell,d}$ we have
$\rho(S_w, M) = O(\ell^{d-1} w)$.
\end{lemma}
\begin{proof}
We can cover the slab $S_w$ with $O((\ell/w)^{d-1})$ $(w,d)$-cubes.  To make this concrete, we cover $\textsf{C}_{\ell,d}$ with $\lceil \ell/w \rceil^d$ cubes on a regular grid.  Then in at least one basis direction (the one closest to orthogonal to the normal of $F$) any column of cubes can intersect $S_w$ in at most $4$ cubes.  Since there are $\lceil \ell/w \rceil^{d-1}$ such columns, the bound holds.  Let $\c{C}_w$ be the set of these cubes covering $S_w$.  

Restricted to any one such cube $C_w$, the contribution of those edges to $\rho(S_w,M)$ is at most $O(w^d)$ by Lemma \ref{lem:match-cube}.  Now we need to argue that we can just sum the effect of all covering cubes.  The concern is that an edge goes through many cubes, only contributing a small amount to each $\rho(C_w,M)$ term, but when the total length is taken to the $d$th power it is much more.  However, since each edge's contribution is capped at $(2w)^2$, we can say that if any edge goes through more than $O(1)$ cubes, its length must be at least $w$, and its contribution in one such cube is already $\Omega(w)$, so we can simply inflate the effect of each cube towards $\rho(S_w,M)$ by a constant.  

In particular, consider any edge $\overline{pq}$ that has $p \in C_w$.  
Each cube has $3^d-1$ neighboring cubes, including through vertex incidence.  Thus if edge $\overline{pq}$ passes through more than $3^d$ cubes, $q$ must be in a cube that is not one of $C_w$'s neighbors.  
Thus it must have length at least $w$; and hence its length in at least one cube $C_w$ must be at least $w/3^d$, with its contribution to $\rho(C_w,M) > w^d/(3^{d^2})$.  Thus we can multiply the effect of each edge in $\rho(C_w,M)$ by $3^{d^2} 2^d = O(1)$ and be sure it is at least as large as the effect of that edge in $\rho(S_w,M)$.  Hence 
\begin{align*}
\rho(S_w,M) &
\leq 3^{d^2} 2^d \sum_{C_w \in \c{C}_w} \rho(C_w,M) \leq O(1) \sum_{C_w \in \c{C}_w} O(w^d) 
\\ &
= O((\ell/w)^{d-1}) \cdot O(w^d) = O(\ell^{d-1} w).  
\end{align*}
\end{proof}

We can apply the same decomposition as used to prove Lemma \ref{lem:match-slab} to also prove a result for a $w$-expanded volume $G_w$ around a degree $g$ polynomial surface $G$.  A degree $g$ polynomial surface can intersect a line at most $g$ times, so for some $\textsf{C}_{\ell,d}$ the expanded surface $G_w \cap \s{C}_{\ell,d}$ can be intersected by $O(g (\ell/w)^{d-1})$ $(w,d)$-cubes.  Hence we can achieve the following bound.  

\begin{corollary}
\label{lem:match-curve}
Let $P \subset \textsf{C}_{\ell,d}$, where $d$ is constant, and let $M$ be its min-cost matching.  
For any volume $G_w$ defined by a polynomial surface of degree $g$ expanded by a width $w$, restricted to $\textsf{C}_{\ell,d}$ we have
$\rho(G_w, M) = O(g \ell^{d-1} w)$.
\end{corollary}

\section{Constructing $\eps$-Samples for Smoothed Range Spaces}
\label{sec:smhalf}

In this section we build on the ideas from \cite{Phillips2013} and the new min-cost matching results in Section \ref{sec:match} to produce new discrepancy-based $\eps$-sample bounds for smoothed range spaces.  
The basic construction is as follows.  We create a min-cost matching $M$ on $P$, then for each pair $(p,q) \in M$, we retain one of the two points at random, halving the point set.  We repeat this until we reach our desired size.  This should not be unfamiliar to readers familiar with discrepancy-based techniques for creating $\eps$-samples of binary range spaces~\cite{Mat99,Cha01}.  In that literature similar methods exist for creating matchings ``with low-crossing number''.  Each such matching formulation is specific to the particular combinatorial range space one is concerned with.  However, in the case of smoothed range spaces, we show that the min-cost matching approach is a universal algorithm.  It means that an $\eps$-sample $Q$ for one smoothed range space $(P,\c{H}_w)$ is also an $\eps$-sample for any other smoothed range space $(P,\c{H}'_w)$, perhaps up to some constant factors.  
We also show how these bounds can sometimes improve upon $\eps$-sample bounds derived from linked range spaces; herein the parameter $w$ will play a critical role.

\subsection{Discrepancy for Smoothed Halfspaces}
\label{sec:m2disc}
To simplify arguments, we first consider $P \subset \b{R}^2$ extending to $\b{R}^d$ in Section \ref{sec:extd}. 

Let $\chi: P \to \{-1,+1\}$ be a coloring of $P$, and define the \emph{discrepancy} of $(P,\c{H}_w)$ with coloring $\chi$ as $\disc_\chi(P,\c{H}_w) = \max_{h \in \c{H}_w} |\sum_{p \in P} \chi(p) v_h(p)|$. Restricted to one smoothed range $h \in \c{H}_w$ this is $\disc_\chi(P,h) = |\sum_{p \in P} \chi(p) v_h(p)|$.  
We construct a coloring $\chi$ using the min-cost matching $M$ of $P$; for each $\{p_i,q_i\} \in M$ we randomly select one of $p_i$ or $q_i$ to have $\chi(p_i) = +1$, and the other $\chi(q_i) = -1$.  
We next establish bounds on the discrepancy of this coloring for a $\varsigma$-bounded smoothed range space $(P,\c{H}_w)$, i.e., where the gradient of $v_h$ is bounded by $\varsigma \leq c_1/w$ for a constant $c_1$ (see Section \ref{sec:smooth-RS}).  

For any smoothed range $h \in \c{H}_w$, we can now define a random variable $X_j = \chi(p_j) v_h(p_j) + \chi(q_j) v_h(q_j)$ for each pair $\{p_j,q_j\}$ in the matching $M$.  This allows us to rewrite $\disc_\chi(P,h) = |\sum_j X_j|$.  We can also define a variable $\Delta_j = 2|v_h(p_j) - v_h(q_j)|$ such that $X_j \in \{-\Delta_j/2, \Delta_j/2\}$.  Now following the key insight from \cite{Phillips2013} we can bound $\sum_j \Delta_j^2$ using results from Section \ref{sec:match}, which shows up in the following Chernoff bound from \cite{DP09}:
Let $\{X_1, X_2, \ldots\}$ be independent random variables with $\textbf{\sffamily E}[X_j]=0$ and $X_j = \{-\Delta_j/2, \Delta_j/2\}$ then 
\begin{equation}\label{eq:CH}
\Pr\Big[\disc_\chi(P,h) \geq \alpha\Big] = \Pr\Big[\Big|\sum_j X_j \Big| \geq \alpha \Big] \leq 2 \exp\left(\frac{-2\alpha^2}{\sum_j \Delta_j^2}\right).
\end{equation}


\begin{lemma}
Assume $P \subset \b{R}^2$ is contained in some cube $\textsf{C}_{\ell,2}$ and with min-cost matching $M$ defining $\chi$, and consider a $\varsigma$-bounded smoothed halfspace $h \in \c{H}_w$ associated with slab $S_w$.  
Let $\rho(S_w,M) \leq c_2(\ell w)$ for constant $c_2$ (see definition of $\rho$ in Section \ref{sec:match}).  
Then $\Pr \bigg[\disc_\chi(P,h) > C\sqrt{\frac{\ell}{w} \log(2/\delta)} \bigg] \leq \delta$ for any $\delta > 0$ and constant $C = c_1\sqrt{2c_2}$.
\label{lem:discK2d}
\end{lemma}

\begin{proof}
Using the gradient of $v_h$ is at most $\varsigma = c_1/w$ and $|v_h(p_j) - v_h(q_j)| \leq \varsigma \max\{2w, \|p_j-q_j\|\}$ we have
\[
\sum_j \Delta_j^2 
= 
\sum_j 4(v_h(p_j) - v_h(q_j))^2 
\leq 
4 \varsigma^2 \rho(S_w,M) 
\leq 
4c_1^2/w^2 \cdot c_2\ell w
=4c_1^2 c_2 \ell/w,
\]
where the second inequality follows by Lemma \ref{lem:match-slab} which shows that $\rho(S_w, M) = \sum_j \max\{(2w)^2, \|p_j - q_j\|^2\} \leq c_2 (\ell w)$.   

We now study the random variable $\disc_\chi(P,h) = |\sum_i X_i|$ for a single $h \in \c{H}_w$.  
Invoking (\ref{eq:CH}) we can bound
$\Pr[\disc_\chi(P,h) > \alpha] \leq 2 \exp(-\alpha^2/(2 c_1^2c_2\ell /w))$.
Setting $C = c_1\sqrt{2c_2}$ and $\alpha = C\sqrt{\frac{\ell}{w} \log(2/\delta)}$  reveals $\Pr \bigg[\disc_\chi(P,h) > C \sqrt{\frac{\ell}{w} \log(2/\delta)} \bigg] \leq \delta$.  
 \end{proof}

\subsection{From a Single Smoothed Halfspace to a Smoothed Range Space}
The above theorems imply small discrepancy for a single smoothed halfspace $h \in \c{H}_w$, but this does not yet imply small discrepancy $\disc_\chi(P,\c{H}_w)$, for all choices of smoothed halfspaces simultaneously.  And in a smoothed range space, the family $\c{H}_w$ is not finite, since even if the same set of points have $v_h(p)=1$, $v_h(p) = 0$, or are in the slab $S_w$, infinitesimal changes of $h$ will change $\sde_P(h)$.  
So in order to bound $\disc_\chi(P,\c{H}_w)$, we will show that there are polynomial in $n$ number of smoothed halfspaces that need to be considered, and then apply a union bound across this set.  

\begin{theorem}
\label{thm:net2}
For $P \subset \b{R}^2$ of size $n$, for $\c{H}_w $, and value $\Psi(n,\delta)=O\bigg(\sqrt{\frac{\ell}{w}\log \frac{n}{\delta}}\bigg)$ for $\delta > 0$, we can choose a coloring $\chi$ such that $\Pr[\disc_\chi(P,\c{H}_w) > \Psi(n,\delta)] \leq \delta$.
\end{theorem}

\begin{proof}
We define a net of smoothed halfspaces $\c{S}_\alpha \subset \c{H}_w$ where any smoothed halfspace $h \in \c{H}_w$ assigns a value $v_h(p)$ to a point $p\in P$, then there always exists a smoothed halfspace $s \in \c{S}_\alpha$ such that $\forall_{p \in P} |v_h(p) - v_s(p)| \leq \alpha \varsigma$. Since there are only $|P| = n$ points, the difference $\sum_{p \in P} |v_h(p) - v_s(p)|$ is no more than $n \alpha \varsigma$ . By setting $\alpha = 1/n\varsigma$ we can ensure that $\disc_\chi(P, h) < \disc_\chi(P,s) +1$. Thus if all $s \in \c{S}_\alpha$ have small discrepancy, then all smoothed halfspaces in $\c{H}_w$ have small discrepancy.  

We now describe a construction of $\c{S}_\alpha$ (illustrated in Figure \ref{fig:stripnet}) of size at most $O(n^4)$ and then apply the union bound in Lemma \ref{lem:discK2d} to only increase the discrepancy in that bound by a $\sqrt{\log n}$ factor. 
First consider the halfspace with boundary passing through each pair of points $p,p' \in P$. For each such halfspace, and for each point ($p$ or $p'$) it passes through, consider $4w/\alpha$ rotations around that point (wlog $p$).  Make the increment of the rotation such that the closest point $p'_F$ on the rotated boundary $F$ increases a distance $\|p' - p'_F\|$ of $\alpha/2$ in each next rotation.  That is, the projection distance $\|p' - p'_F\|$ on each rotation around $p$ is a distance of $\alpha/2, \alpha, 3\alpha/2, ..., 2w$; this is repeated in each direction.  
Now, for each rotated halfspace, consider $4w/\alpha$ translations in the direction normal to the halfspace.  There are $2w/\alpha$ translations in the normal direction, and its opposite, at increments of $\alpha/2$ (e.g., $\alpha/2$, $\alpha$, $3\alpha/2$, ... $2w$). 

Since $\alpha = 1/n\varsigma$ and $w = O(1/\varsigma)$, then $4w/\alpha =O(n)$. 
Thus the size of $\c{S}_\alpha$ is $O(n^4)$: for each of $O(n^2)$ pairs, there are $O(n)$ rotations and for each rotations there are $O(n)$ translations.  

We now show for any $h \in \c{H}_w$ how to map to the smoothed halfspace in $s \in \c{S}_\alpha$ such that for all $p \in P$ that $|v_h(p) - v_s(p)| \leq \alpha \varsigma$.  
First consider all points $P \cap S_w$, where $S_w$ is the slab defined by $h$.  If the slab is empty then the closest two points $p,p' \in P$ would generate one translation and rotation $s \in S_\alpha$ that moved both of them out of the slab, causing all of the same values $v_h(p_i) = v_s(p_i) \in \{0,1\}$.  
Otherwise, for any point $p$ in the slab, there exists some rotation moving $p_F$ by at most $\alpha/2$ and another rotation moving $p_F$ by at most $\alpha/2$ resulting in $|v_h(p) - v_s(p)| \leq \|p-p_F\| \cdot \varsigma \leq (\alpha/2+\alpha/2) \cdot \varsigma = \alpha\varsigma$.  However, we need to ensure this holds for all points simultaneously.  The translations affect $\|p-p_F\|$ for all points the same (at most $\alpha/2$), but the rotations can affect further away points by more.  Thus, we choose the two points $p,p' \in S_\alpha$ that maximize $\|p_F - p'_F\|$, and consider the closest rotation of $h$ to one of the smoothed halfspaces $s \in \c{S}_\alpha$ that they generate.  The rotation will affect all other points less than it will those two, and thus at most $\alpha/2$, as desired.  

Finally we set the probability of failure in Lemma \ref{lem:discK2d} as $\delta' = \Omega(\delta/|\c{S}_\alpha|)$ for each smoothed halfspace.  This implies that for $\Psi(n,\delta)=C\sqrt{\frac{\ell}{w}\log(2/\delta'}) = O\bigg(\sqrt{\frac{\ell}{w}\log \frac{n}{\delta}}\bigg)$, the $\Pr[\disc_\chi(P,\c{H}_w) > \Psi(n,\delta)] \leq \delta$.
\end{proof}

\begin{figure}[ht!]
\begin{center}   \includegraphics[width=1.5in]{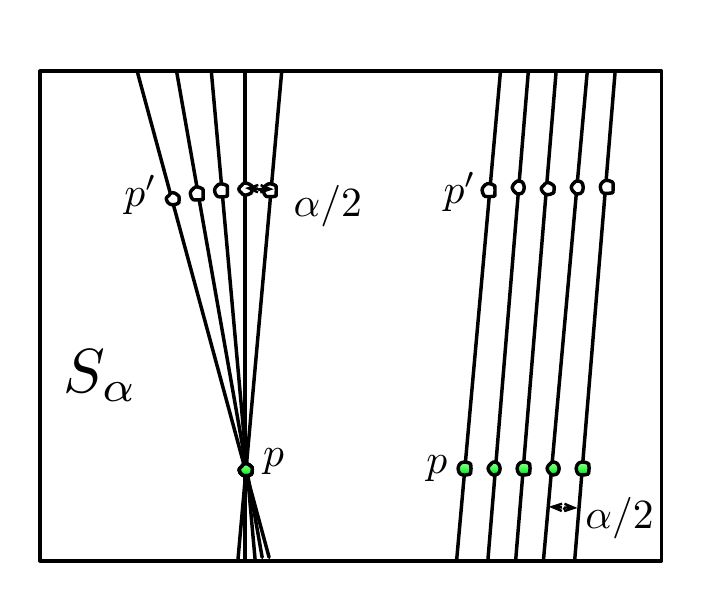}
\end{center}
    \caption{Illustration of the shifts and rotations to create $S_\alpha$. }
  \label{fig:stripnet}
\end{figure}

\subsection{$\eps$-Samples for Smoothed Halfspaces}
To transform this discrepancy algorithm to $\eps$-samples, let $f(n) = \disc_\chi (P, \c{H}_w)/n$ be the value of $\eps$ in the $\eps$-samples generated by a single coloring of a set of size $n$. Solving for $n$ in terms of $\eps$, the sample size is $s(\eps) = O(\frac{1}{\eps} \sqrt{\frac{\ell}{w} \log\frac{\ell}{w\eps\delta}})$.  We can then apply the \emph{MergeReduce} framework~\cite{CM96}; iteratively apply this random coloring in $O(\log n)$ rounds on disjoint subsets of size $O(s(\eps))$.  Using a generalized analysis (c.f., Theorem 3.1 in \cite{phillips2008algorithms}), we have the same $\eps$-sample size bound.

\begin{theorem}
\label{thm:epssample}
For $P \subset \s{C}_{\ell,2} \subset \b{R}^2$, with probability at least $1-\delta$, we can construct an $\eps$-sample of $(P, \c{H}_w)$ of size $O\Big(\frac{1}{\eps} \sqrt{\frac{\ell}{w} \log\frac{\ell}{w\eps\delta}}\Big)$.
\end{theorem}

To see that these bounds make rough sense, consider a random point set $P$ in a unit square (so $\ell=1$).  
Then setting $w = 1/n$ will yield roughly $O(1)$ points in the slab (and should roughly revert to the non-smoothed setting); this leads to $\disc_\chi(P,\c{H}_w) = O(\sqrt{n} \sqrt{\log(n/\delta)})$ and an $\eps$-sample of size $O((1/\eps^2)\sqrt{\log(1/\eps\delta)})$, basically the random sampling bound.  
But setting $w = \eps$ so about $\eps n$ points are in the slab (the same amount of error we allow in an $\eps$-sample) yields $\disc_\chi(P,\c{H}_w) = O((1/\sqrt{\eps n}) \cdot \sqrt{\log(n/\delta)})$ and the size of the $\eps$-sample to be $O(\frac{1}{\eps} \sqrt{\log(1/\eps\delta)})$, which is a large improvement over $O(1/\eps^{4/3})$, and the best bound known for non-smoothed range spaces~\cite{Mat99}.  

\subsection{Adaptive Bounds for Non-Uniform Distributed Data}
\label{app:adaptive}
However, the assumption that $P \subset \s{C}_{\ell,2}$ (although not uncommon~\cite{Mat99}) can be restrictive.  In this section, we attempt to relax this assumption.  We do not see how to completely remove some such assumption using our suite of techniques since it could be all of the data lies very close to a line $l$, and then a halfspace boundary similar to that line $l$ will have all of the points within the slab.  In this case, we should not expect much better than with binary range spaces unless we make $w$ much larger than the average deviation of points from the line $l$.  

However, we can do better, if the data is ``well-clustered''.  That is, consider partitioning the data into subsets $P_1, P_2, \ldots, P_k$ so that each $P_i$ is contained in an $(\ell_k,2)$-cube.   Then we can replace $\ell$ in the previous bound with a value $\Phi_k = \max\{k-1, k \cdot \ell_k\}$.  In particular, let $\Phi = \min_{k \geq 1} \max \Phi_k$.  
We can then bound the contribution of each $(\ell_k,2)$ cube towards $\rho(S_w,M)$ as $O(\ell w)$ using Lemma \ref{lem:match-slab}, and the sum of them as $\Phi_k = \max\{k-1, k \cdot \ell_k\}$ since there will at most $k-1$ edges between these $k$ boxes.  In Lemma \ref{lem:discK2d} this yields $\Pr\Big[\disc_\chi(P,h) = O(\sqrt{(\Phi/w) \log(1/\delta)})\Big] \leq \delta$, and eventually with probability $1-\delta$ an $\eps$-sample of size $O\big((1/\eps) \sqrt{(\Phi/w) \log(1/\eps \delta)}\big)$ in place of Theorem \ref{thm:epssample}.  

\begin{theorem}
\label{thm:phi-epssample}
Consider a partition of $P = \bigcup_i P_i$ for $P \subset \b{R}^2$, so each $P_i$ is in a $(\ell_k,2)$-cube, and letting $\Phi_k = \max\{k-1, k \cdot \ell_k\}$ and $\Phi = \min_{k \geq 1} \Phi_k$.  
Then with probability at least $1-\delta$, we can construct an $\eps$-sample of $(P, \c{H}_w)$ of size $O\Big(\frac{1}{\eps} \sqrt{\frac{\Phi}{w} \log\frac{\Phi}{w\eps\delta}}\Big)$.
\end{theorem}

\vspace{-2mm}

We can compute a $2\sqrt{2}$-approximation to $\Phi$ in $O(nk_{\max}^2)$ time, where $k_{\max}$ is the largest value of $k$ we consider ($k_{\max} = \log n$ may be a good choice).  Our algorithm will only use axis-aligned cubes which is a $\sqrt{2}$-approximation to more generally allowing rotated cubes to fit each $P_i$.  We simply run the $k$-clustering of \cite{Gon85} using the $L_\infty$ metric.  That is, we start with an arbitrary point $p\in P$ to place in a set $W_1$; this represents the center of the smallest $(\ell,2)$-cube that fits all data.  Then we inductively, choose $p_k = \arg \max_{p \in P} \min_{w \in W_{k-1}} \|w-p\|_\infty$, and create $W_k = W_{k-1} \cup p_k$.  At any stage $\ell_k = 2 \cdot \max_{p \in P} \min_{w \in W_{k-1}} \|w-p\|_\infty$, and $\Phi_k = \max\{k-1,\ell_k \cdot k\}$.  

\paragraph{Non-linear clusters.}
We can also observe a slightly tighter bound.  If there are $k$ $(\ell_k,2)$-cubes, but they are not all near a single line, then they cannot all contribute to the discrepancy.  Given a partition $P = \bigcap_i P_i$ where each $P_i$ is in a $(\ell_k, 2)$-cube, let $\kappa_k$ describe the maximum number of these cubes that a single slab $S_w$ can intersect.  Then we can use $\bar \Phi_k = \max\{k-1, \kappa_k \cdot \ell_k\}$ in place of $\Phi_k$.  
However, it is less clear the best way to construct an approximation to $\bar \Phi_k$ and $\bar \Phi = \min_{k\geq1} \bar \Phi_k$.  

As another thought experiment, consider all of the points are in $\s{C}_{\ell,2}$.  
We can now decompose this square into $w^2$ smaller squares, each of side length $\ell/w$.  
Any slab $S_w$ can only pass through $O(w)$ smaller squares; thus $\bar \Phi_{w^2} = \min\{w-1,O(w) \cdot \ell/w\} = O(\ell)$.  So we recover the original non-adaptive bound.  

One may wonder if this can be improved if many of the $w^2$ squares are empty.  
If there are $O(1)$ non-empty squares, then $\Phi$ already captures this improved bound.  
If there is still a slab $S_w$ pass through $\Omega(w)$ squares, then this $\bar \Phi$ bound again does not improve over the non-adaptive one.  
However if there are $\Theta(w)$ non-empty squares, and no slab $S_w$ passes through more than $O(1)$ of them (e.g., they are all on the boundary of $\textsf{C}_{\ell,2}$), then $\bar \Phi$ improves the bound over $\Phi$ by a factor of $w$.  Thus the $\bar \Phi$ approach can improve the bound in certain settings.  

\subsection{Generalization to $d$ Dimensions}
\label{sec:extd}

We now extend from $\b{R}^2$ to $\b{R}^d$ for $d > 2$.  
Using results from Section \ref{sec:match} we implicitly get a bound on $\sum_j \Delta_j^d$, but the Chernoff bound we use requires a bound on $\sum_j \Delta_j^2$.  
As in \cite{Phillips2013}, we can attain a weaker bound using Jensen's inequality over at most $n$ terms
\begin{equation}
\left(\sum_j \frac{1}{n} \Delta_j^2\right)^{d/2} \leq \sum_j \frac{1}{n} \left(\Delta_j^2\right)^{d/2}
\;\; \text{ so } \;\;
\sum_j \Delta_j^2 \leq n^{1-2/d} \left(\sum_j \Delta_j^d\right)^{2/d}.
\label{eq:Jensen}
\end{equation}
Replacing this bound and using $\rho(S_w,M) \leq O(\ell^{d-1} w)$ in Lemma \ref{lem:discK2d} and considering $\varsigma = c_1/w$ for some constant $c_1$ results:

\begin{lemma}
\label{lem:discKd}
Assume $P \subset \b{R}^d$ is contained in some cube $\textsf{C}_{\ell,d}$ and with min-cost matching $M$, and consider a $\varsigma$-bounded smoothed halfspace $h \in \c{H}_w$ associated with slab $S_w$.  Let $\rho(S_w,M) \leq c_2(\ell^{d-1} w)$ for constant $c_2$.  
Then $\Pr \Big[\disc_\chi(P,h) > C n^{1/2 - 1/d} (\ell/w)^{1-1/d}  \sqrt{ \log(2/\delta)} \Big] \leq \delta$ for any $\delta > 0$, where $C = \sqrt{2}c_1(c_2)^{1/d} $ is a constant.  
\end{lemma}

\begin{proof}
Using the gradient of $v_h$ is at most $\varsigma = c_1/w$ and $|v_h(p_j) - v_h(q_j)| \leq \varsigma \max\{2w, \|p_j-q_j\|\}$ we have
\[
\sum_j \Delta_j^d 
= 
\sum_j 2^d(v_h(p_j) - v_h(q_j))^d 
\leq 
2^d \varsigma^d \rho(S_w,M) 
\leq 
2^dc_1^d/w^d \cdot c_2\ell^{d-1} w
=2^dc_1^d c_2 (\ell/w)^{d-1},
\]
where second inequality follows by Lemma \ref{lem:match-slab} that $\rho(S_w, M) = \sum_j \max\{(2w)^d, \|p_j - q_j\|^d\} \leq c_2 (\ell^{d-1} w)$. Hence, by Jensen's inequality (i.e. (\ref{eq:Jensen}))
\[\sum_j \Delta_j^2 \leq n^{1-2/d} (2^d c_1^dc_2(\ell/w)^{d-1})^{2/d} = n^{1-2/d} 4c_1^2(c_2)^{2/d} (\ell/w)^{2(d-1)/d}.
\]  
We now study the random variable $\disc_\chi(P,h) = |\sum_i X_i|$ for a single $h \in \c{H}_w$.  
Invoking (\ref{eq:CH}) we can bound
\[
\Pr[\disc_\chi(P,h) > \alpha] \leq 2 \exp(-\alpha^2/n^{1-2/d} 2c_1^2(c_2)^{2/d} (\ell/w)^{2(d-1)/d}.
\] 
Setting $C = \sqrt{2}c_1(c_2)^{1/d} $ and 
\[
\alpha = C n^{1/2 - 1/d}  (\ell/w)^{1-1/d}  \sqrt{ \log(2/\delta)}
\] 
reveals 
\[
\Pr \Big[\disc_\chi(P,h) > C n^{1/2 - 1/d} (\ell/w)^{1-1/d}  \sqrt{ \log(2/\delta)}\Big] \leq \delta. 
\]  
\end{proof}

For all choices of smoothed halfspaces, applying the union bound, the discrepancy is increased by a $\sqrt{\log n}$ factor, with the following probabilistic guarantee, 
\[\Pr[\disc_\chi(P,\c{H}_w) > C n^{1/2 - 1/d} (\ell/w)^{1-1/d}  \sqrt{ \log(n/\delta)}] \leq \delta.\] 
Ultimately, we can extend Theorem \ref{thm:epssample} to the following.  
\begin{theorem}
\label{thm:epssample-d}
For $P \subset \s{C}_{\ell,d} \subset \b{R}^d$, where $d$ is constant, with probability at least $1-\delta$, we can construct an $\eps$-sample of $(P, \c{H}_w)$ of size 
$O\Big((\ell/w)^{2(d-1)/(d+2)} \cdot \Big(\frac{1}{\eps} \sqrt{\log\frac{\ell}{w\eps\delta}}\Big)^{2d/(d+2)}\Big)$.
\end{theorem}

If the data is ``well-clustered'' in high dimension, we can get a similar adaptive bounds as Theorem \ref{thm:phi-epssample}.  

\begin{theorem}
\label{thm:phi-epssample-d}
Consider a partition of $P = \bigcup_i P_i$ for $P \subset \b{R}^d$, so each $P_i$ is in a $(\ell_k,d)$-cube, and letting $\Phi_k = \max\{k-1, k \cdot \ell_k\}$ and $\Phi = \min_{k \geq 1} \Phi_k$.  
Then with probability at least $1-\delta$, we can construct an $\eps$-sample of $(P, \c{H}_w)$ of size $O\Big((\Phi/w)^{2(d-1)/(d+2)} \cdot \Big(\frac{1}{\eps} \sqrt{\log\frac{\Phi}{w\eps\delta}}\Big)^{2d/(d+2)}\Big)$.
\end{theorem} 

Note these results address scenario (S3) from the introduction where we want to find a small set (the $\eps$-sample) so that it could be much smaller than the $d/\eps^2$ random sampling bound, and allows generalization error $O(\eps)$ for agnostic learning as described in Section \ref{sec:e-sample-SRS}.  
When $\ell/w$ (or $\Phi/w$) is constant, the exponents on $1/\eps$ are also better than those for binary ranges spaces (see Section \ref{sec:definitions}).

\bibliographystyle{abbrv}
\bibliography{kernel-refs-short}

\end{document}